\documentclass[11pt]{article}
\usepackage{amsmath, amsthm, amssymb, graphicx, color, xcolor, array}

\newtheorem{thm}{Theorem}[section]

\newtheorem{lem}[thm]{Lemma}

\newcommand{\prob}{\operatorname{Prob}}

\newcommand{\Z}{\mathbb{Z}}

\newcommand{\pat}{\textnormal{Pat}}
\newcommand{\nonc}{\textrm{NC}}

\newcommand{\leftplaq}{\ell}
\newcommand{\rightplaq}{r}
\newcommand{\rows}{\mathcal{R}_L}

\newcommand{\leftsymb}{\mathord{\leftarrow}}
\newcommand{\rightsymb}{\mathord{\rightarrow}}
\newcommand{\centersymb}{\mathord{*}}
\newcommand{\spac}{\ \,}

\newcommand{\pplaquette}[1]{\!\!\raisebox{-2.5pt}{
\setlength{\unitlength}{0.25pt}
\begin{picture}(50,50)(0,0)
\put(0,0){\framebox(50,50){#1}}
\end{picture}
}}
\newcommand{\auxplaquette}[1]{\!\!\!
\raisebox{-6pt}{
\setlength{\unitlength}{0.2pt}
\begin{picture}(50,100)(0,0)
\put(0,0){\framebox(50,100){#1}}
\end{picture}
}}

\title{Bijective combinatorial proof of the commutation of transfer matrices in the dense $O(1)$ loop model}

\author{Ron Peled \and Dan Romik}

\begin{document}

\maketitle

\begin{abstract}
The dense $O(1)$ loop model is a statistical physics model with connections to the quantum XXZ spin chain, alternating sign matrices, the six-vertex model and critical bond percolation on the square lattice. When cylindrical boundary conditions are imposed, the model possesses a commuting family of transfer matrices. The original proof of the commutation property is algebraic and is based on the Yang-Baxter equation. In this paper we give a new proof of this fact using a direct combinatorial bijection.
\end{abstract}

\renewcommand{\thefootnote}{\fnsymbol{footnote}}
\renewcommand{\thefootnote}{\arabic{footnote}}

\section{Introduction}

In the \textbf{dense $O(1)$ loop model}, a square lattice is tiled with the following two kinds of square tiles known as \textbf{plaquettes}, denoted symbolically by $\leftplaq$ and $\rightplaq$:
\begin{center}
\begin{tabular}{ccc}
\raisebox{22pt}{$\leftplaq:=$}
\scalebox{0.17}{\includegraphics{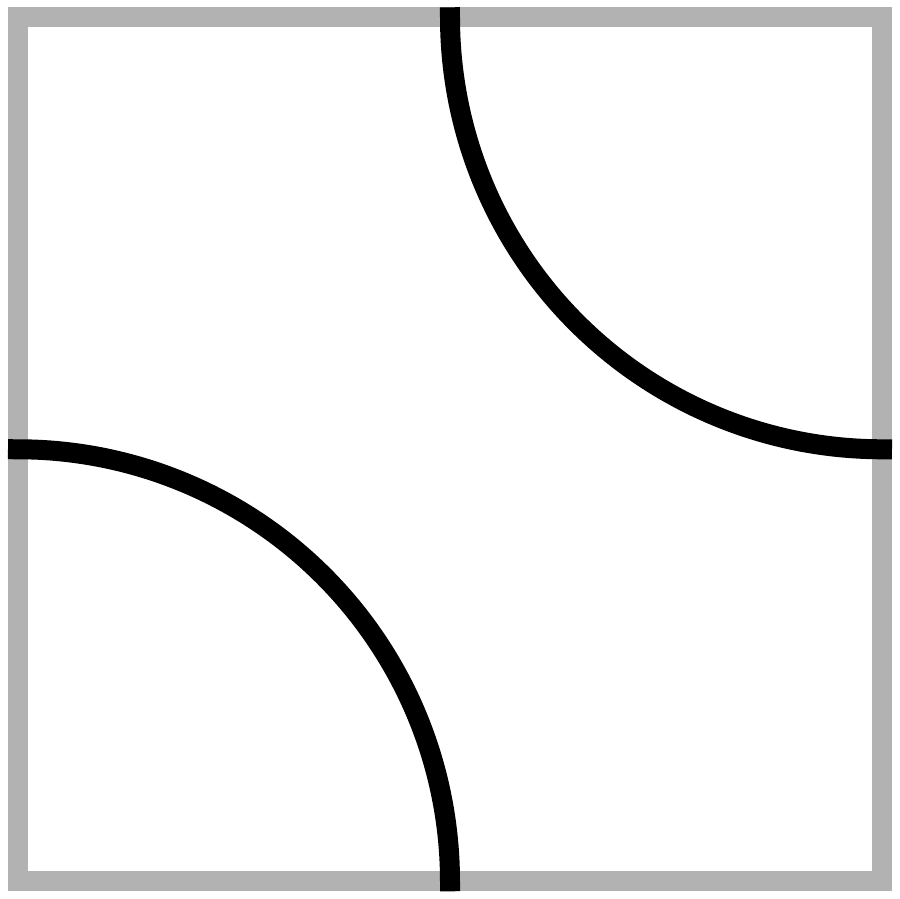}} & &
\raisebox{22pt}{$\rightplaq:=$}
\scalebox{0.17}{\includegraphics{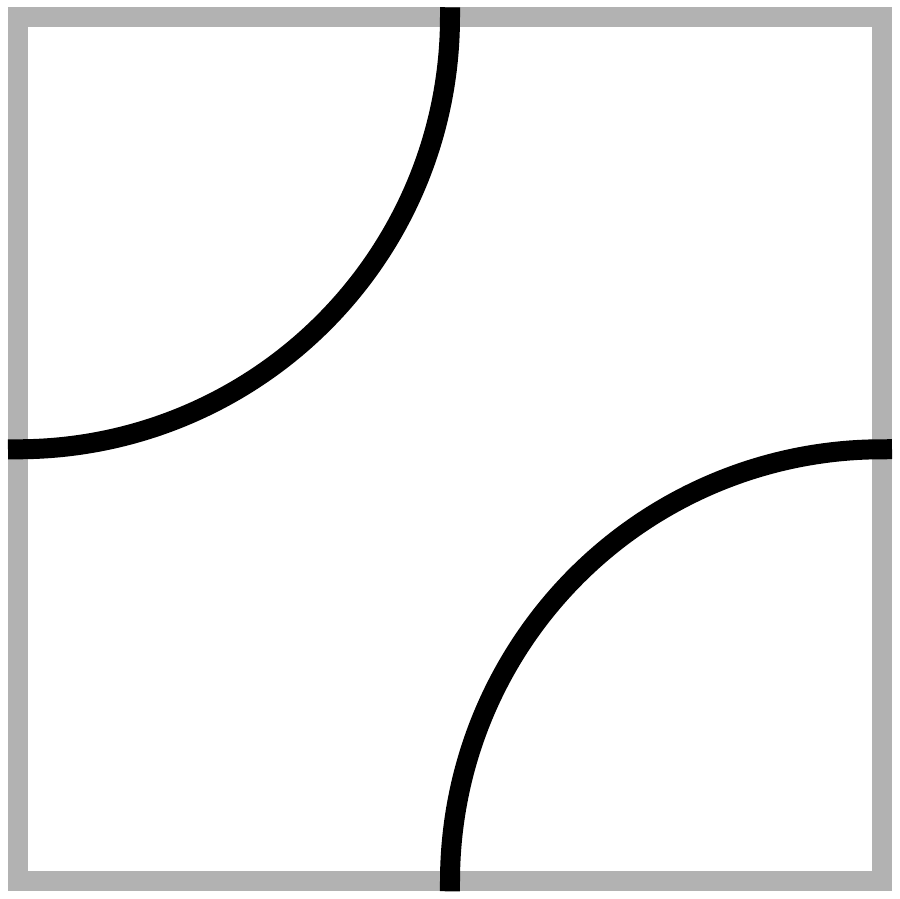}}
\end{tabular}
\end{center}

An arrangement of plaquettes induces a \textbf{connectivity pattern} on points, referred to as \textbf{endpoints}, that lie mid-edge on boundary edges of the configuration. We will focus on a particular case studied in several recent papers \cite{zinn-justin-di-francesco1, fonseca-zinn-justin, romik, zinn-justin-di-francesco2}, which concerns a semi-infinite cylindrical geometry obtained by considering a plaquette tiling on the strip $[0,L]\times[0,\infty)$ with periodic boundary conditions along the \hbox{$x$-direction}, where $L:=2n$ is an even integer. In this case, the connectivity pattern is a \textbf{noncrossing matching} (also sometimes called a \textbf{link pattern} \cite{zinn-justin}) of the $2n$ boundary endpoints; see Figure~\ref{fig:noncrossing}.

\begin{figure}[h]
\begin{center}
\begin{tabular}{ccc}
\scalebox{0.8}{\includegraphics{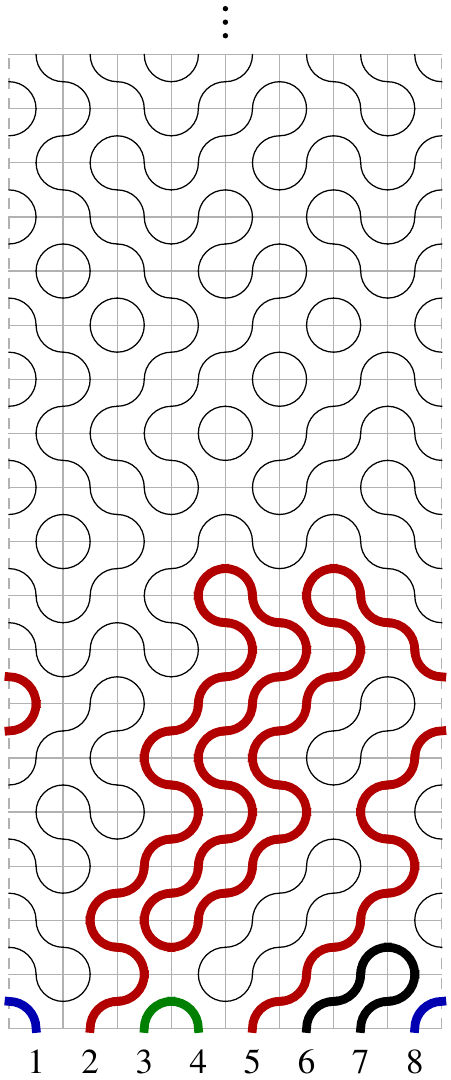}} & \quad
\raisebox{100pt}{$\longrightarrow$} & \raisebox{113pt}{
\begin{tabular}{cc}
\scalebox{0.58}{\includegraphics{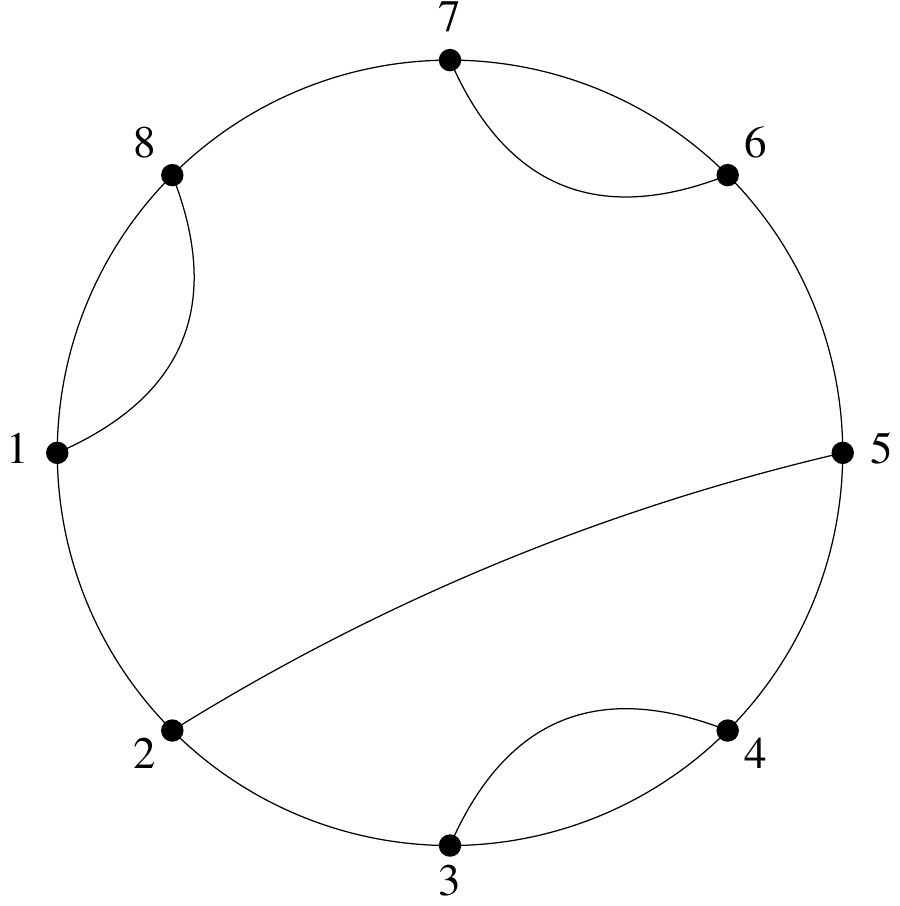}}
\\[10pt]
\scalebox{0.65}{\includegraphics{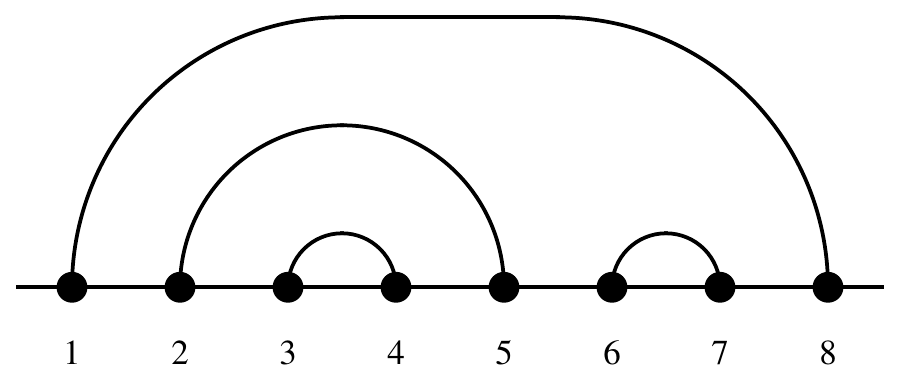}}
\end{tabular}
}
\\[10pt]
(a) & & (b)
\end{tabular}
\caption{
(a) An arrangement of plaquettes on a semi-infinite vertical strip. The left and right boundary edges are identified to form a topological cylinder, and paths connecting pairs of endpoints are highlighted. (b) The induced noncrossing matching of $2n=8$ endpoints, shown in two equivalent representations as a matching of points around a circle or on a line.
}
\label{fig:noncrossing}
\end{center}
\end{figure}

Denote the set of noncrossing matchings of $2n$ points by $\nonc_n$. Equipping the set of plaquette tilings with a probability measure induces a probability distribution on $\nonc_n$. In the simplest case in which plaquettes are chosen independently and with equal probabilities for both plaquette types, this model was referred to in \cite{romik} as \textbf{loop percolation}, and is essentially a way to encode information about connectivities in critical bond percolation on $\Z^2$. The distribution of the induced noncrossing matching in this case has remarkable properties and is closely related to the \textbf{XXZ~spin~chain} from quantum statistical mechanics and to the \textbf{fully packed loops} model associated with alternating sign matrices and the six-vertex model; see \cite{romik} for recent results and many references.

A natural generalization of the $\textrm{Bernoulli}(1/2)$ measure on
plaquettes discussed above is a so-called inhomogeneous version in
which plaquettes are sampled independently but with respective
probabilities $p,1-p$ for the two types of plaquettes $\leftplaq=\
$\raisebox{-3pt}{\scalebox{0.05}{\includegraphics{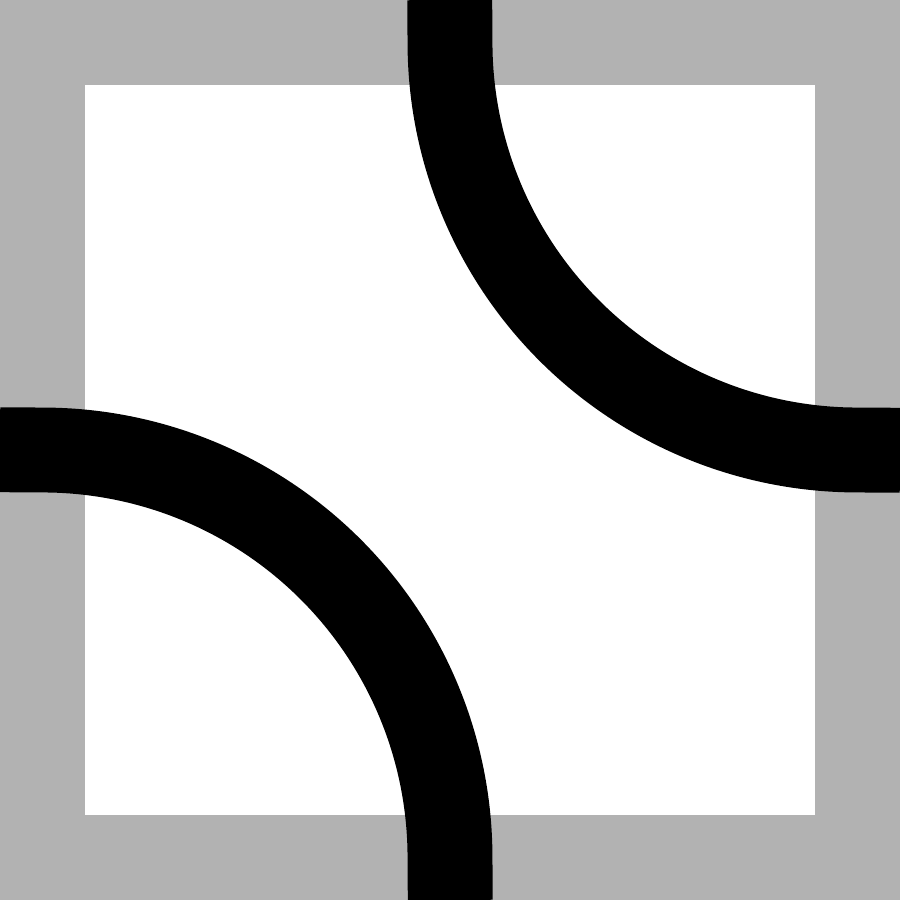}}}
and $\rightplaq=\
$\raisebox{-3pt}{\scalebox{0.05}{\includegraphics{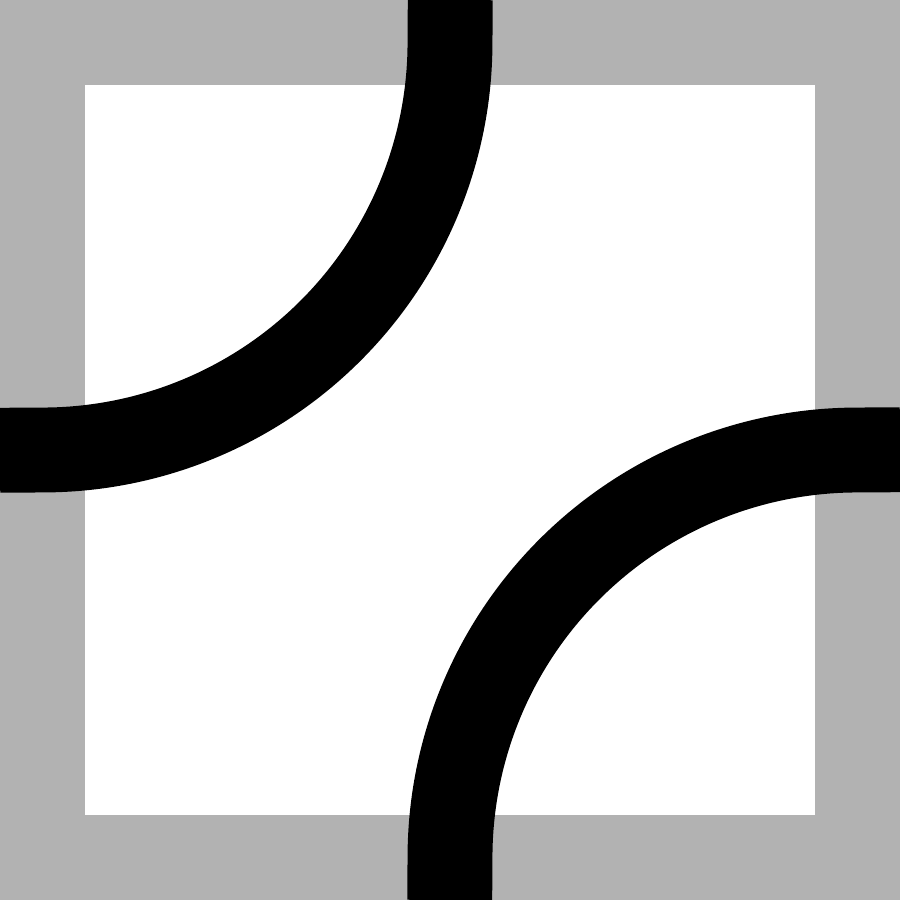}}}
(call a plaquette sampled randomly in this way a \textbf{$p$-biased
plaquette}). The value of $p$ can even depend on the position of the
plaquette. In particular, two variants of such an inhomogeneous
model play an important role in the theory: in one variant, we
assign different biases $p_1,\ldots,p_L$ to plaquettes in different
columns. The probability distribution of the induced connectivity
pattern in this inhomogeneous model involves functions of the
parameters $p_1,\ldots,p_L$ which have a natural algebraic
description as multivariate polynomials, known as \textbf{wheel
polynomials}, which are solutions to the \textbf{quantum
Knizhnik-Zamolodchikov (qKZ) equation}
\cite{zinn-justin-di-francesco2}.

The second type of inhomogeneous model involves the assignment of different biases $p_1,p_2,\ldots$ to plaquettes in different \emph{rows}, that is, the plaquettes in the $j$th row (counting from the bottom) are sampled with bias $p_j$. In this case, we have the following surprising invariance result.

\begin{thm}[Invariance of the connectivity pattern distribution] \label{thm:invariance}
Let $(p_j)_{j=1}^\infty$ be a sequence of numbers in $[0,1]$ such that $\sum_{j=1}^\infty p_j^n(1-p_j)^n=\infty$. When the plaquettes are selected independently at random with bias $p_j$ for plaquettes in the $j$th row as described above, the resulting connectivity pattern is almost surely well-defined, and its distribution is independent of the biases~$p_j$.
\end{thm}

Theorem~\ref{thm:invariance} is an easy consequence (see Section~\ref{sec:final-remarks}) of a more fundamental algebraic fact regarding the commutation of a family of row transfer matrices associated with the model, which encode the effect on the connectivity pattern of an added row of $p$-biased plaquettes.
To make this precise, we first define a way for a row $\rho \in \{\leftplaq, \rightplaq\}^L$ of plaquettes to act on a noncrossing matching $\pi\in \nonc_n$ and produce a new matching $\rho(\pi)$: this is done by graphically ``composing'' $\rho$ and $\pi$, that is, drawing the row $\rho$ below the diagram associated with $\pi$ and ``pulling the strings''; see Figure~\ref{fig:graphical-composition} for an example. Next, for each $p\in[0,1]$, we define a matrix $T_L^{(p)}= (t_{\pi,\pi'}^{(p)})_{\pi,\pi'\in \nonc_n}$ whose rows and columns are indexed by noncrossing matchings of order~$n$. The entries of $T_L^{(p)}$ are defined as the transition probabilities
$$t_{\pi,\pi'}^{(p)} := \prob( \pi'=\rho_p(\pi) ), $$
where $\rho_p$ denotes a random row of independently sampled $p$-biased plaquettes.

\begin{figure}
\begin{center}
\begin{tabular}{ccc}
\scalebox{0.6}{\includegraphics{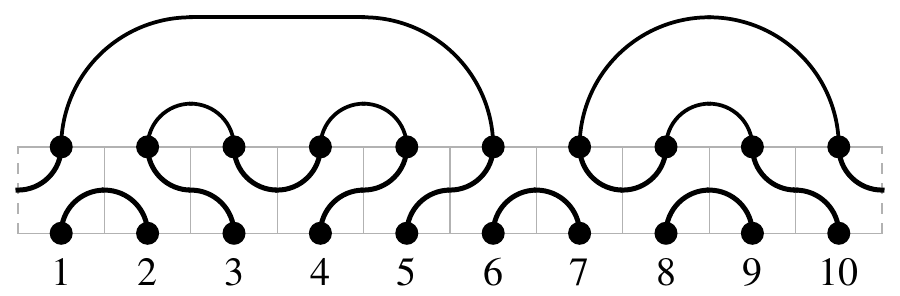}} &
\raisebox{16pt}{$\longrightarrow$} &
\scalebox{0.6}{\includegraphics{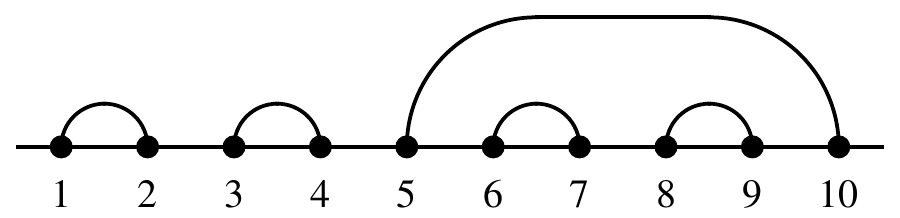}}
\end{tabular}
\caption{Composing a noncrossing matching $\pi$ with a row $\rho$ of plaquettes produces a new matching $\pi'=\rho(\pi)$.}
\label{fig:graphical-composition}
\end{center}

\end{figure}

\begin{thm}[Commutation of transfer matrices] \label{thm:matrices-commute}
The matrices \linebreak
$(T_L^{(p)})_{0\le p\le 1}$ form a commuting family of matrices. That is, for all $p,q\in [0,1]$ we have
$$ T_L^{(p)} T_L^{(q)} = T_L^{(q)} T_L^{(p)}. $$
\end{thm}

Theorem~\ref{thm:matrices-commute} seems to be well-known to experts in the field, but we are unsure of its precise origin. It is mentioned in \cite{mitra-etal} and in Section~3.3.2 of \cite{zinn-justin}, where it is said to follow from the Yang-Baxter equation. We review the idea behind this elegant algebraic technique in Section~\ref{sec:yang-baxter}. Our main goal in this paper is to give a new and more direct proof of the commutation property. In fact, we will show that Theorem~\ref{thm:matrices-commute} follows easily from the existence of a certain combinatorial bijection involving pairs of plaquette rows.

Denote by $\rows:=\{\leftplaq,\rightplaq\}^L$ the set of rows of $L$ plaquettes (each row still being thought of as being arranged around a cylinder). We denote elements of $\rows^2$ by $\binom{\rho_2}{\rho_1}$, where $\rho_1,\rho_2\in \rows$, and think of this as a $2$-row circular arrangement of plaquettes in which $\rho_2$ is placed above $\rho_1$. Denote by $\pat\binom{\rho_2}{\rho_1}$ the connectivity pattern of the $2L=4n$ endpoints on the top and bottom sides of the arrangement.

\begin{thm}[Pattern-preserving involution]
\label{thm:involution}
There exists a map $V: \rows^2 \to \rows^2$ with the following properties:
\begin{enumerate}
\item $V$ is an involution: $V\circ V = \textnormal{Id}$.
\item $V$ is pattern-preserving: $\pat\circ V = \pat$.
\item $V$ switches the numbers of plaquettes of each type between the two rows; that is,
if $\rho_1$ has $j$ plaquettes of type $\leftplaq$ and $\rho_2$ has $k$ plaquettes of type $\leftplaq$, and $\binom{\rho_2'}{\rho_1'}=V\binom{\rho_2}{\rho_1}$, then
$\rho_1'$ has $k$ plaquettes of type $\leftplaq$ and $\rho_2'$ has $j$ plaquettes of type $\leftplaq$.
\item $V$ respects the rotational symmetry of the model, that is, it satisfies $V\circ R = R \circ V$, where $R$ is the operator that rotates pairs of rows by one plaquette in the counterclockwise direction.
\end{enumerate}

\end{thm}

In the next section we explain the construction of the involution $V$ and prove that it satisfies the desired properties. In Section~\ref{sec:proof-commutation} we use Theorem~\ref{thm:involution} to prove Theorem~\ref{thm:matrices-commute}. In Section~\ref{sec:yang-baxter} we compare our approach to the proof of Theorem~\ref{thm:matrices-commute} based on the Yang-Baxter equation.
Section~\ref{sec:final-remarks} has additional remarks.

\section{Construction of the involution}

\label{sec:construction}

To construct the involution $V$, we first define \textbf{block operations}, which are operations that can be performed on a part of a row pair, and which under certain conditions preserve the connectivity pattern. Given numbers $a,b\in \{1,\ldots,L\}$, $a\neq b$, denote by $[a,b]$ the discrete interval of positions from $a$ to $b$, where in the case $a>b$ this is interpreted in the sense of circle arithmetic, that is, $[a,b]=\{a,a+1,\ldots,L,1,2,\ldots,b\}$. Given a row pair $\binom{\rho_2}{\rho_1}=\left(\begin{smallmatrix} y_1 & y_2 & \ldots & y_L \\ x_1 & x_2 & \ldots & x_L \end{smallmatrix}\right)$, the block operator $B_{[a,b]}$ associated with the interval $[a,b]$ transforms $\binom{\rho_2}{\rho_1}$ by rotating the contiguous block of columns indexed by the numbers in $[a,b]$ by 180 degrees around its center. In the case $a<b$ this gives
$$ B_{[a,b]}\binom{\rho_2}{\rho_1} =
\left(\begin{array}{ccccccccccc}
y_1 & y_2  & \ldots & y_{a-1} & x_b & x_{b-1} & \ldots & x_a & y_{b+1} & \ldots & y_L  \\
x_1 & x_2  & \ldots & x_{a-1} & y_b & y_{b-1} & \ldots & y_a & x_{b+1} & \ldots & x_L
\end{array} \right),
$$
and the case $a>b$ is analogous.
Note that for disjoint intervals $[a,b],[c,d]$, the operators $B_{[a,b]}$ and $B_{[c,d]}$ commute.

Say that an interval $[a,b]$ is a \textbf{fundamental interval} of the row pair $\binom{\rho_2}{\rho_1}=\left(\begin{smallmatrix} y_1 & y_2 & \ldots & y_L \\ x_1 & x_2 & \ldots & x_L \end{smallmatrix}\right)$ if the block
$\binom{\rho_2}{\rho_1}_{\raisebox{2pt}{\big|}[a,b]}$ is of the form
$$
\raisebox{-7pt}{\textrm{$
\begin{pmatrix} y_a & y_{a+1} & \ldots & y_b \\ x_a & x_{a+1} & \ldots & x_b \end{pmatrix}
=$}
}
\raisebox{-8pt}{\bigg(} \begin{array}{c}
\overbrace{ \vphantom{\leftplaq\atop\rightplaq} {\rightplaq\atop\leftplaq} \ \ {\rightplaq\atop\leftplaq} \ \ {\cdots \atop \cdots}\ \  {\rightplaq\atop\leftplaq} }^{j\textrm{ columns}}
\ \ \displaystyle {\beta \atop \alpha} \ \
\overbrace{ {\leftplaq\atop\rightplaq} \ \ {\leftplaq\atop\rightplaq} \ \  {\cdots \atop \cdots} \ \  {\leftplaq\atop\rightplaq} }^{k\textrm{ columns}}
\end{array}
\raisebox{-8pt}{\bigg)}
$$
for some $\alpha,\beta\in\{\leftplaq,\rightplaq\}$ and $j,k\ge 0$ such that $j+k>0$. In this case the block
$\binom{\rho_2}{\rho_1}_{\raisebox{2pt}{\big|}[a,b]}$ is called a \textbf{fundamental block}
of $\binom{\rho_2}{\rho_1}$.

\begin{lem}
\label{lem:block-pattern-inv}
If $\binom{\rho_2}{\rho_1}_{\raisebox{2pt}{\big|}[a,b]}$ is a fundamental block for the row pair $\binom{\rho_2}{\rho_1}$, and its columns are not all equal, then $B_{[a,b]}\binom{\rho_2}{\rho_1}$ has the same connectivity pattern as $\binom{\rho_2}{\rho_1}$.
\end{lem}

\begin{proof}
It suffices to check that the connectivity pattern of the $2(j+k+1)+4$ endpoints around the four (horizontal \emph{and} vertical) boundary edges of the block $[a,b]$) remains unchanged by the block operation $B_{[a,b]}$. For example, in the case $j=2,k=4$, we have the picture
\begin{center}
\begin{tabular}{ccc}
\scalebox{0.6}{\includegraphics{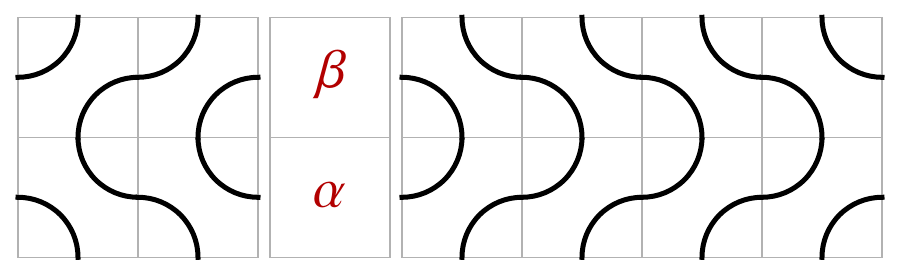}} &
\raisebox{21pt}{$\longrightarrow$} &
\scalebox{0.6}{\includegraphics{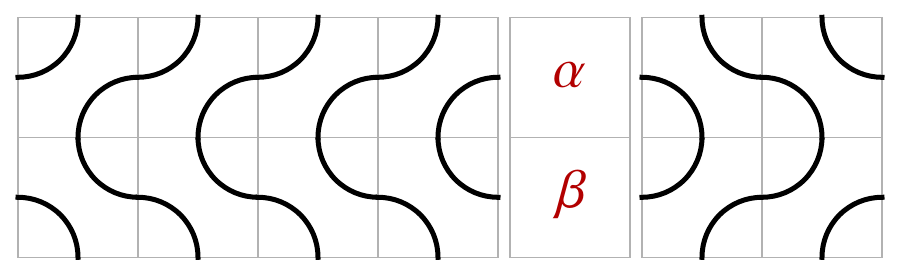}}
\end{tabular}
\end{center}
where there are four choices for $\binom{\beta}{\alpha}$. When $\binom{\beta}{\alpha}=\binom{\rightplaq}{\rightplaq}$ this becomes
\begin{center}
\begin{tabular}{ccc}
\scalebox{0.6}{\includegraphics{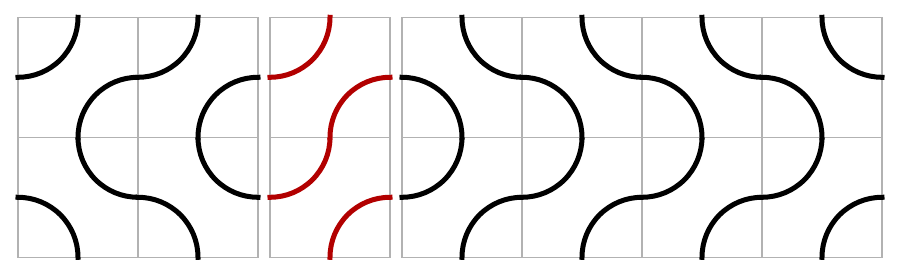}} &
\raisebox{21pt}{$\longrightarrow$} &
\scalebox{0.6}{\includegraphics{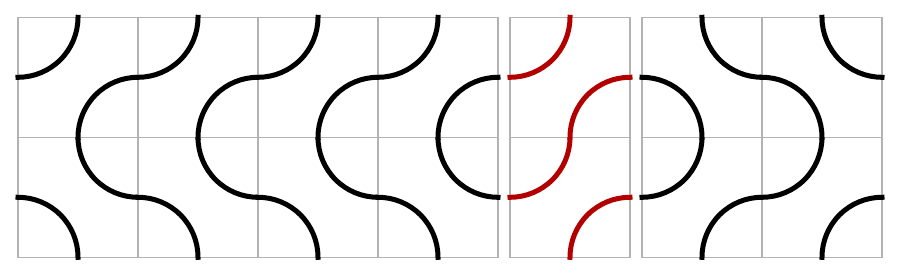}}
\end{tabular}
\end{center}
and we see that the connectivities of the $18$ boundary endpoints are indeed unchanged. Similarly, when $\binom{\beta}{\alpha}=\binom{\leftplaq}{\rightplaq}$
the picture is
\begin{center}
\begin{tabular}{ccc}
\scalebox{0.6}{\includegraphics{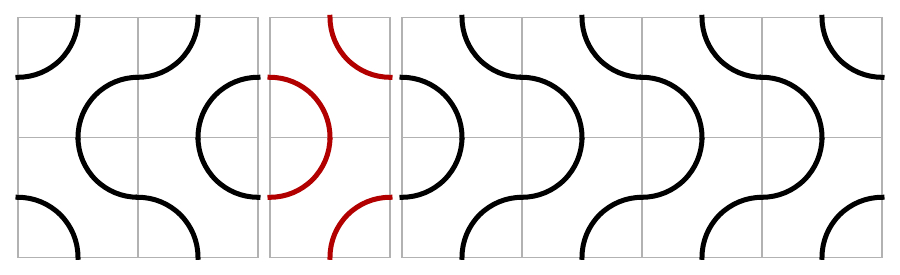}} &
\raisebox{21pt}{$\longrightarrow$} &
\scalebox{0.6}{\includegraphics{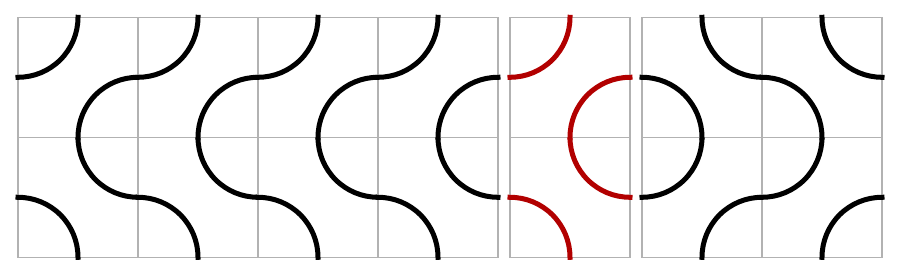}}
\end{tabular}
\end{center}
and the claim is also satisfied. The remaining two cases
$\binom{\beta}{\alpha}=\binom{\leftplaq}{\leftplaq}, \binom{\rightplaq}{\leftplaq}$
are similarly easy to verify. While this illustrates the claim for the specific values $j=2,k=4$, it is clear that the same argument applies generally for any values $j,k\ge 1$:
for any of the possible choices for $\alpha,\beta$, the connectivity pattern is in this case represented schematically by the diagram
\begin{center}
\scalebox{0.6}{\includegraphics{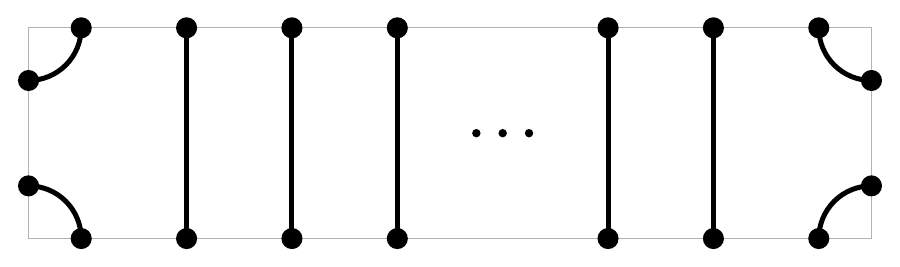}}
\end{center}
and is clearly invariant under the operation of rotating the block by 180 degrees around its center. Finally, when one of $j$ or $k$ is $0$, the connectivity pattern is still represented by the schematic diagram above, but one needs to check this separately, noting that the assumption of the lemma requires that $\binom{\beta}{\alpha} \neq \binom{\leftplaq}{\rightplaq}$ in the case when $j=0$, or
$\binom{\beta}{\alpha} \neq \binom{\rightplaq}{\leftplaq}$ in the case when $k=0$. We omit this easy verification.
\end{proof}

We are now ready to define the involution $V$. The idea is to apply the block operator $B_{[a,b]}$ for any fundamental interval of $\binom{\rho_2}{\rho_1}$ that is \emph{maximal} (with respect to containment). First, define $V$ explicitly for two special configurations by setting
\begin{align*}
V\begin{pmatrix}
\rightplaq & \rightplaq & \cdots & \rightplaq \\
\leftplaq & \leftplaq & \cdots & \leftplaq
\end{pmatrix}
&:=
\begin{pmatrix}
\leftplaq & \leftplaq & \cdots & \leftplaq \\
\rightplaq & \rightplaq & \cdots & \rightplaq
\end{pmatrix},
\\[3pt]
V\begin{pmatrix}
\leftplaq & \leftplaq & \cdots & \leftplaq \\
\rightplaq & \rightplaq & \cdots & \rightplaq
\end{pmatrix}
&:=
\begin{pmatrix}
\rightplaq & \rightplaq & \cdots & \rightplaq \\
\leftplaq & \leftplaq & \cdots & \leftplaq
\end{pmatrix}.
\end{align*}
Next, for any row pair $\binom{\rho_2}{\rho_1}
 \neq
\begin{pmatrix}
\rightplaq & \rightplaq & \cdots & \rightplaq \\
\leftplaq & \leftplaq & \cdots & \leftplaq
\end{pmatrix}, \begin{pmatrix}
\leftplaq & \leftplaq & \cdots & \leftplaq \\
\rightplaq & \rightplaq & \cdots & \rightplaq
\end{pmatrix}$ from among the remaining $4^L-2$ possibilities, let $[a_1,b_1], \ldots, [a_k,b_k]$ be the maximal fundamental intervals of the pair. We tentatively define
$$ V\binom{\rho_2}{\rho_1} := B_{[a_1,b_1]} \cdots B_{[a_k,b_k]} \binom{\rho_2}{\rho_1}. $$
The next lemma implies that $V$ is well-defined.

\begin{lem}
\label{lem:blocks-disjoint}
Given a row pair
$\binom{\rho_2}{\rho_1}
 \neq
\begin{pmatrix}
\rightplaq & \rightplaq & \cdots & \rightplaq \\
\leftplaq & \leftplaq & \cdots & \leftplaq
\end{pmatrix}, \begin{pmatrix}
\leftplaq & \leftplaq & \cdots & \leftplaq \\
\rightplaq & \rightplaq & \cdots & \rightplaq
\end{pmatrix}$,
any two maximal fundamental intervals are disjoint, and in the case when there's only one maximal fundamental interval $[a,b]$ such that $[a,b]=[1,L]$, the description of the interval as a fundamental interval is unique (that is, $a$ and $b$ are uniquely determined).
\end{lem}

Note that the second claim of the lemma is needed to resolve a possible ambiguity in the definition of $V$ in the case when the entire cycle $[1,L]$ is a fundamental interval. The first claim implies that the order in which the block operators $B_{[a_j,b_j]}$ are applied is unimportant.

\begin{proof} For visualization purposes, it is convenient to think of the lemma as a symbolic claim about circular strings of letters over the alphabet $\{ \leftarrow, \rightarrow, * \}$, where the letters in this alphabet correspond to columns of two plaquettes according to the rule
$$
 \leftarrow\  := \binom{\leftplaq}{\rightplaq}, \qquad \rightarrow\  := \binom{\rightplaq}{\leftplaq}, \qquad *\, := \binom{\leftplaq}{\leftplaq} \textrm{ or }\binom{\rightplaq}{\rightplaq}.
$$
With this schematic representation, it is easy to see that the rules for finding the maximal fundamental intervals in a circular string $s \in \{ \leftarrow, \rightarrow, * \}^L$ can be reformulated as follows:
\begin{enumerate}
\item A symbol $x$ is part of a maximal fundamental block if and only if it is a ``$\leftarrow$'', a ``$\rightarrow$'', or it is a ``$*$'' with a ``$\leftarrow$'' to its right or a ``$\rightarrow$'' to its left.
\item Any symbol $x$ to the right of a symbol $y=\textrm{``}\rightarrow\textrm{''}$ belongs to the maximal fundamental block that contains $y$.
\item Any symbol $x$ to the left of a symbol $y=\textrm{``}\leftarrow\textrm{''}$ belongs to the maximal fundamental block that contains $y$.
\end{enumerate}
It is now easy to prove the disjointness claim: starting from a symbol $x$ that forms part of a maximal fundamental block, one can uniquely identify the maximal fundamental block to which it belongs by successively moving to the left and growing the block if it is allowed by the rules until one can go no further, and then doing the same to grow the block to the right from $x$.
As a small example, starting from the string
$$
\centersymb
\spac
\centersymb
\spac
\rightsymb
\spac
\rightsymb
\spac
\rightsymb
\spac
\leftsymb
\spac
*
\spac
*
\spac
\rightsymb
\spac
*
\spac
*
\spac
\rightsymb
\spac
\rightsymb
\spac
*
\spac
\leftsymb
\spac
\leftsymb
\spac
*
\spac
\leftsymb
\spac
*
\spac
*
\spac
*
\spac
\leftsymb
\spac
\rightsymb
\spac
*
$$
 and applying the above procedure, we find the following maximal fundamental blocks:
$$
\centersymb
\spac
\centersymb
\spac
\textrm{\fbox{$
\rightsymb
\spac
\rightsymb
\spac
\rightsymb
\spac
\leftsymb
$}}
\spac
*
\spac
*
\spac
\textrm{\fbox{$
\rightsymb
\spac
*
$}}
\spac
*
\spac
\textrm{\fbox{$
\rightsymb
\spac
\rightsymb
\spac
*
\spac
\leftsymb
\spac
\leftsymb
$}}
\spac
\textrm{\fbox{$
*
\spac
\leftsymb
$}}
\spac
*
\spac
*
\spac
\textrm{\fbox{$
*
\spac
\leftsymb
$}}
\spac
\textrm{\fbox{$
\rightsymb
\spac
*
$}}
$$
The second uniqueness claim that pertains to the case when one maximal fundamental block spans the entire cycle $[1,L]$ follows from similar reasoning.
\end{proof}

\begin{proof}[Proof of Theorem~\ref{thm:involution}]
Having defined the mapping $V$, we now prove that it satisfies the claims of the Theorem. For the two special cases for which $V$ was defined separately the claim is trivial, so we focus on the remaining cases $\binom{\rho_2}{\rho_1}
 \neq \textrm{``}\rightarrow\  \rightarrow \ldots\rightarrow \textrm{''}, \textrm{``}\leftarrow\  \leftarrow \ldots\leftarrow \textrm{''}$
 (using the notation from the proof of Lemma~\ref{lem:blocks-disjoint}). Note also that the definition of $V$ automatically satisfies the rotational equivariance property 4 of the Theorem, so it remains to prove properties 1--3.

We start by claiming that the maximal fundamental intervals in $V\binom{\rho_2}{\rho_1}$ are the same as those for $\binom{\rho_2}{\rho_1}$. To see this, it is enough to show that each block rotation operation $B_{[a_j,b_j]}$ preserves the set of maximal fundamental intervals. We verify this in the case when there is no maximal fundamental block spanning the entire cycle $[1,L]$, and leave that case as an exercise to the reader. First, let us check that $[a_j,b_j]$ is still a maximal fundamental interval after the rotation. Indeed, the definition of a fundamental interval is symmetric under the operation of rotating the block 180 degrees around its center, so after the operation $[a_j,b_j]$ is still a fundamental interval; its maximality is ensured by the fact that the letter in position $a_j-1$ (which was not changed by the operation $B_{[a_j,b_j]}$) is not a ``$\rightarrow$'', and the letter in position $b_j+1$ is not a ``$\leftarrow$'', so the block cannot grow in either direction and still remain fundamental. Next, we check the claim for the remaining intervals. They of course preserve their fundamental property (since they were untouched by the rotation of $[a_j,b_j]$), and if for any such interval $[a_k,b_k]$ the rotation of $[a_j,b_j]$ caused it to lose its maximality (which is only a theoretical possibility for the intervals adjacent to $[a_j,b_j]$) that must mean that the maximal fundamental interval containing $[a_k,b_k]$ now intersects $[a_j,b_j]$, in contradiction to the disjointness claim of Lemma~\ref{lem:blocks-disjoint}. So the claim is proved.

As a consequence of the above claim, we immediately get the result that $V$ is an involution. Next, note that, by the list of rules formulated in the proof of Lemma~\ref{lem:blocks-disjoint}, a maximal fundamental block cannot consist of just ``$\leftarrow$'' symbols or of just ``$\rightarrow$'' symbols,
unless it spans the entire cycle $[1,L]$, but we have specifically excluded those two cases. It follows that the maximal fundamental blocks satisfy the assumption of Lemma~\ref{lem:block-pattern-inv}, and therefore each of the block operations $B_{[a_j,b_j]}$ preserves the connectivity pattern; hence $V$ does so as well. This is the second property of $V$ claimed in the theorem.

Finally, the third claimed property concerning the effect $V$ has on the numbers of plaquettes of type $\leftplaq$ and $\rightplaq$ in the top and bottom rows is straightforward, since clearly the switching of the numbers of plaquettes of type $\leftplaq$ between the rows occurs within each maximal fundamental block, and the only columns that are not affected by block operations (that is, that do not belong to a fundamental block) are of the form $\binom{\leftplaq}{\leftplaq}$ or $\binom{\rightplaq}{\rightplaq}$, and therefore contribute an equal number of plaquettes of type $\leftplaq$ to both rows.
\end{proof}

\section{Proof of the commutation of transfer matrices}

\label{sec:proof-commutation}

Theorem~\ref{thm:involution} easily implies the commutation property in Theorem~\ref{thm:matrices-commute}, as follows: let $p,q \in [0,1]$. Let $\rho_p$, $\rho_q$ denote two independently chosen random rows of plaquettes, with $\rho_p$ being a row of $p$-biased plaquettes and $\rho_q$ being a row of $q$-biased plaquettes. By standard properties of Markov transition matrices, for each $\pi,\pi'\in\nonc_n$, the $(\pi,\pi')$-entry of the product
$T_L^{(p)} T_L^{(q)}$ of transfer matrices can be expressed as
$$
\left(T_L^{(p)} T_L^{(q)}\right)_{\pi,\pi'} = \prob( \pi' = \rho_q ( \rho_p (\pi)) )
=
\prob\left( \pi' = \binom{\rho_p}{\rho_q} (\pi) \right),
$$
where in the second equality we generalize the notation $\rho(\pi)$ defined in the Introduction to include a more general action $C(\pi)$ representing the graphical composition (in the same sense as that of Figure~\ref{fig:graphical-composition}) of a cylindrical column $C$ consisting of several rows of plaquettes (in the case above, $C=\binom{\rho_p}{\rho_q}$ has two rows) and a noncrossing matching $\pi\in\nonc_n$. We can now use the involution $V$. By its property of being pattern-preserving, we have
$\binom{\rho_p}{\rho_q} (\pi)=\left(V\binom{\rho_p}{\rho_q} \right)(\pi)$. But by the fact that $V$ is an involution and its switching effect on the numbers of plaquettes of either type between the top and bottom rows,
the row pair $V\binom{\rho_p}{\rho_q}$ is equal in distribution to $\binom{\rho_q}{\rho_p}$, so we can finally write that
\begin{align*}
\left(T_L^{(p)} T_L^{(q)}\right)_{\pi,\pi'}
&= \prob\left( \pi' = \left(V\binom{\rho_p}{\rho_q} \right)(\pi) \right)
\\ &= \prob\left( \pi' = \binom{\rho_q}{\rho_p} (\pi) \right)
=
\left(T_L^{(q)} T_L^{(p)}\right)_{\pi,\pi'},
\end{align*}
which proves the result.
\qed

\bigskip
In the proof above we argued using probabilistic language. A more explicit version of the same argument proceeds by writing the matrix coefficients $t_{\pi,\pi'}^{(p)}$ of $T_L^{(p)}$ explicitly as
$$
t_{\pi,\pi'}^{(p)} = \sum_{\rho \in \{\leftplaq,\rightplaq\}^L \atop \pi'=\rho(\pi)} p^{\nu_\leftplaq(\rho)} (1-p)^{\nu_\rightplaq(\rho)},
$$
where we denote $\nu_x(\rho) = \#\{ 1\le j\le L\,:\, \rho_j = x \}$. One then has that
\begin{align*}
(T_L^{(p)} T_L^{(q)})_{\pi,\pi'} &= \sum_{\tau \in \nonc_n} t_{\pi,\tau}^{(p)} t_{\tau,\pi'}^{(q)}
\\ &=
\sum_{\tau \in \nonc_n} \sum_{\rho_2\in \{\leftplaq,\rightplaq\}^L \atop \tau=\rho_2(\pi)}
\sum_{\rho_1\in \{\leftplaq,\rightplaq\}^L \atop \pi'=\rho_1(\tau)}
p^{\nu_\leftplaq(\rho_2)} (1-p)^{\nu_\rightplaq(\rho_2)}
q^{\nu_\leftplaq(\rho_1)} (1-q)^{\nu_\rightplaq(\rho_1)}
\\ &=
\sum_{\binom{\rho_2}{\rho_1} \in \rows^2, \  \pi'=\binom{\rho_2}{\rho_1}(\pi)}
p^{\nu_\leftplaq(\rho_2)} (1-p)^{\nu_\rightplaq(\rho_2)}
q^{\nu_\leftplaq(\rho_1)} (1-q)^{\nu_\rightplaq(\rho_1)},
\end{align*}
and similarly
$$
(T_L^{(q)} T_L^{(p)})_{\pi,\pi'} =
\sum_{\binom{\rho_2'}{\rho_1'} \in \rows^2,\  \pi'=\binom{\rho_2'}{\rho_1'}(\pi)}
q^{\nu_\leftplaq(\rho_2')} (1-q)^{\nu_\rightplaq(\rho_2')}
p^{\nu_\leftplaq(\rho_1')} (1-p)^{\nu_\rightplaq(\rho_1')}.
$$
The fact that the above two sums are equal follows immediately from Theorem~\ref{thm:involution} by making the substitution $\binom{\rho_2'}{\rho_1'} = V \binom{\rho_2}{\rho_1}$.

\section{Comparison with the Yang-Baxter approach}

\label{sec:yang-baxter}

It is instructive to compare our combinatorial approach to the proof of the commutation of the transfer matrices based on the Yang-Baxter equation. Recall that the Yang-Baxter equation was the name given by Fadeev and his collaborators \cite{fadeev, fadeevetal} to a family of algebraic relations (which are also sometimes referred to as star-triangle relations) that appeared in the study of lattice statistical physics models and in certain other contexts related, e.g., to knot theory; see the surveys \cite{jimbo, perk-auyang} for more details.

The use of the Yang-Baxter algebraic technique to prove the commutation of a family of transfer matrices is standard (see \cite{baxter}), but let us review a version of this argument tailored to our setting and consider it in the light of our new approach to the proof of Theorem~\ref{thm:matrices-commute}. For example, one might wonder if our pattern-preserving involution $V$ is somehow implicit in (or equivalent to) the algebraic manipulations of the Yang-Baxter argument. We argue that this is not the case, that is, that our proof provides a genuinely new approach to understanding the commutation property, and moreover, that the algebraic proof is somewhat weaker, in the sense that there is no clear way to reformulate it in terms of a combinatorial bijection.

The first step in the application of the technique is to formulate a version of the Yang-Baxter equation for the model. We use the notation~\pplaquette{$p$} to denote a $p$-biased random plaquette (sampled independently of all other plaquettes under discussion), and also introduce the notion of an \textbf{auxiliary plaquette}, denoted by \auxplaquette{$s$} and defined as a random $1\times 2$ rectangle taking the possible values
\begin{center}
\raisebox{-22pt}{\scalebox{0.16}{\includegraphics{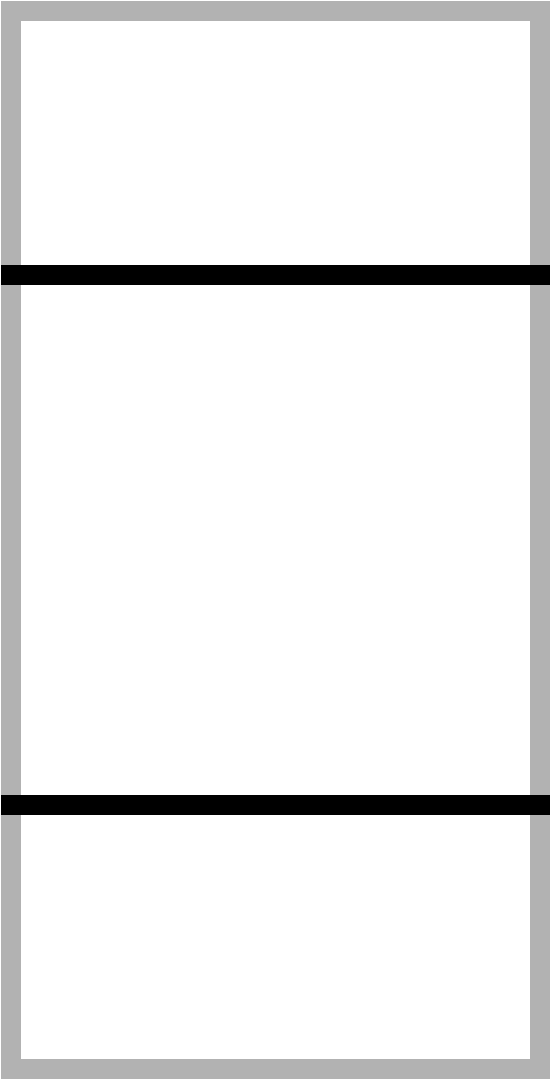}}} \ \
(with prob.\ $s$), \qquad
\raisebox{-22pt}{\scalebox{0.16}{\includegraphics{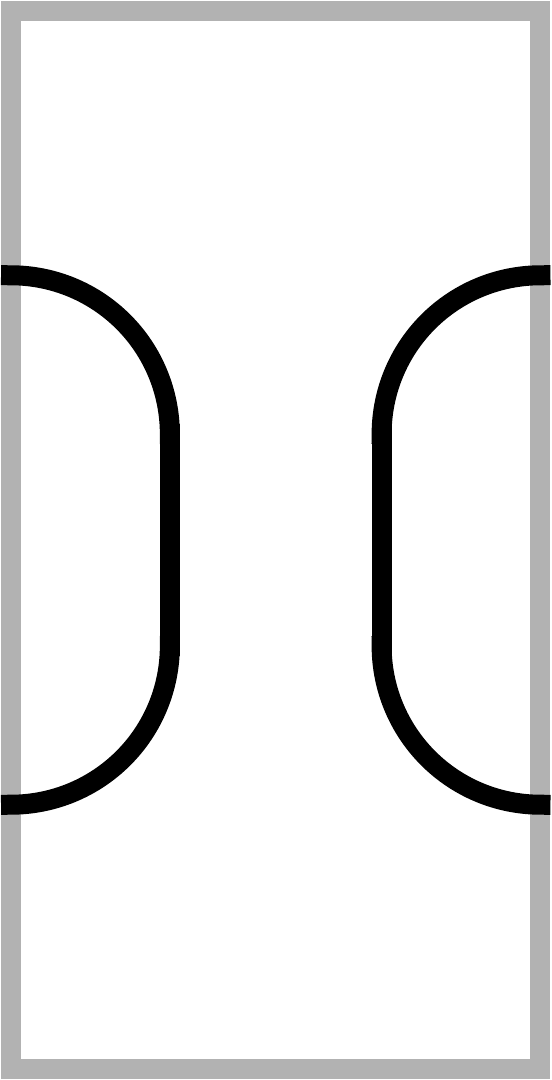}}} \ \
(with prob.\ $1-s$).
\end{center}
\medskip

\begin{lem}
\label{lem:version-yang-baxter}
For any probabilities $0\le p\le q \le 1$ there exists a probability $s=s(p,q)$ such that the connectivity patterns of the boundary points in the two arrangements of random plaquettes
\begin{center}
\setlength{\unitlength}{0.4pt}
\begin{picture}(100,140)(0,-20)
\put(0,0){\framebox(50,100){$s$}}
\put(50,0){\framebox(50,50){$q$}}
\put(50,50){\framebox(50,50){$p$}}
\put(0,25){\circle*{8}}
\put(0,75){\circle*{8}}
\put(75,0){\circle*{8}}
\put(75,100){\circle*{8}}
\put(100,25){\circle*{8}}
\put(100,75){\circle*{8}}
\put(-20,20){$5$}
\put(-20,70){$4$}
\put(70,-25){$6$}
\put(70,110){$3$}
\put(110,20){$1$}
\put(110,70){$2$}
\end{picture}
\quad\qquad \raisebox{34pt}{and}
\quad\qquad
\begin{picture}(100,140)(0,-20)
\put(50,0){\framebox(50,100){$s$}}
\put(0,0){\framebox(50,50){$p$}}
\put(0,50){\framebox(50,50){$q$}}
\put(0,25){\circle*{8}}
\put(0,75){\circle*{8}}
\put(25,0){\circle*{8}}
\put(25,100){\circle*{8}}
\put(100,25){\circle*{8}}
\put(100,75){\circle*{8}}
\put(-20,20){$5$}
\put(-20,70){$4$}
\put(20,-25){$6$}
\put(20,110){$3$}
\put(110,20){$1$}
\put(110,70){$2$}
\end{picture}
\end{center}
are equal in distribution.
\end{lem}

\begin{proof}
For each of the two arrangements of random plaquettes, the eight possibilities for the three plaquettes and their probabilities induce a measure on the five connectivity patterns
\begin{center}
\begin{tabular}{ccccc}
\scalebox{0.18}{\includegraphics{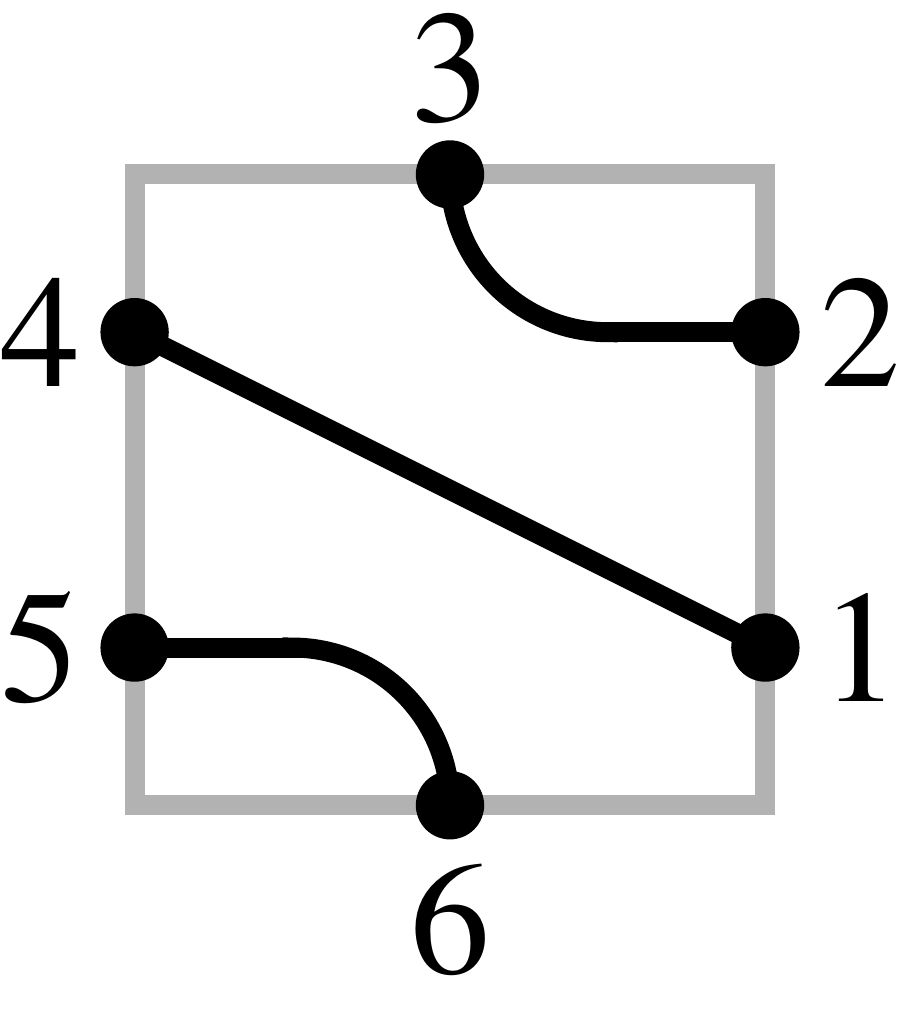}} &
\scalebox{0.18}{\includegraphics{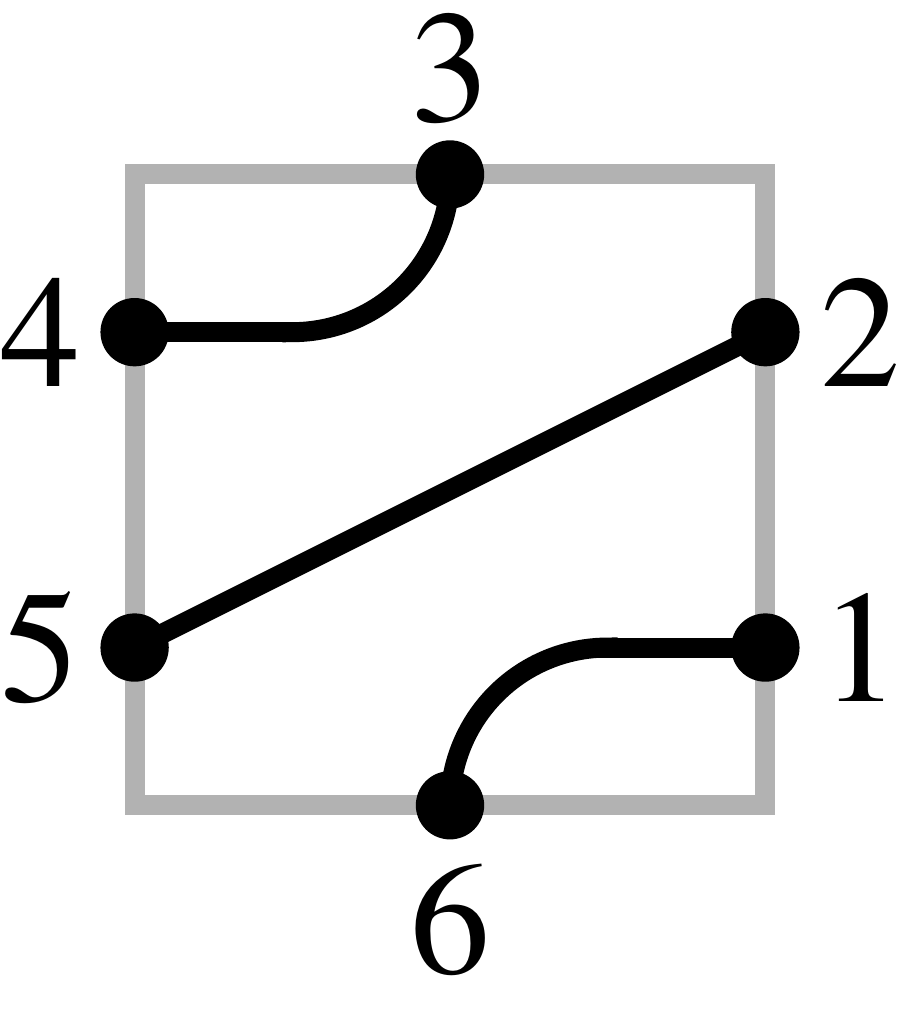}} &
\scalebox{0.18}{\includegraphics{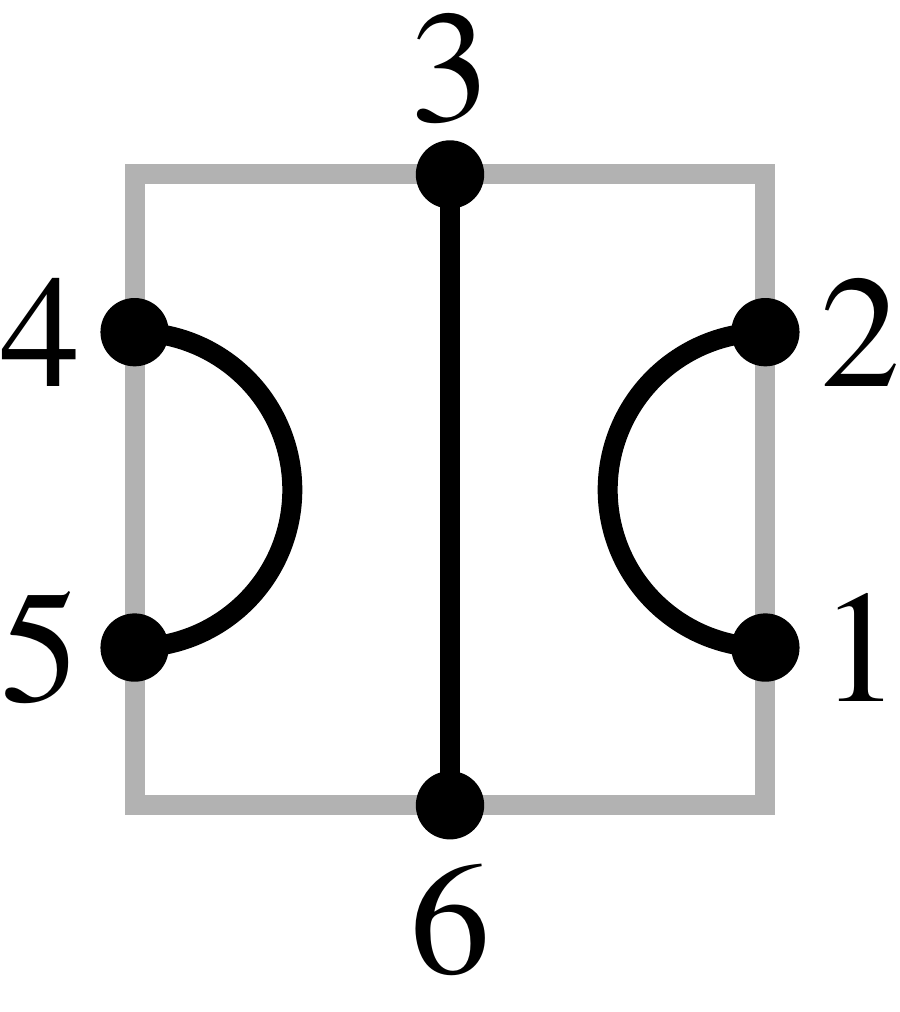}} &
\scalebox{0.18}{\includegraphics{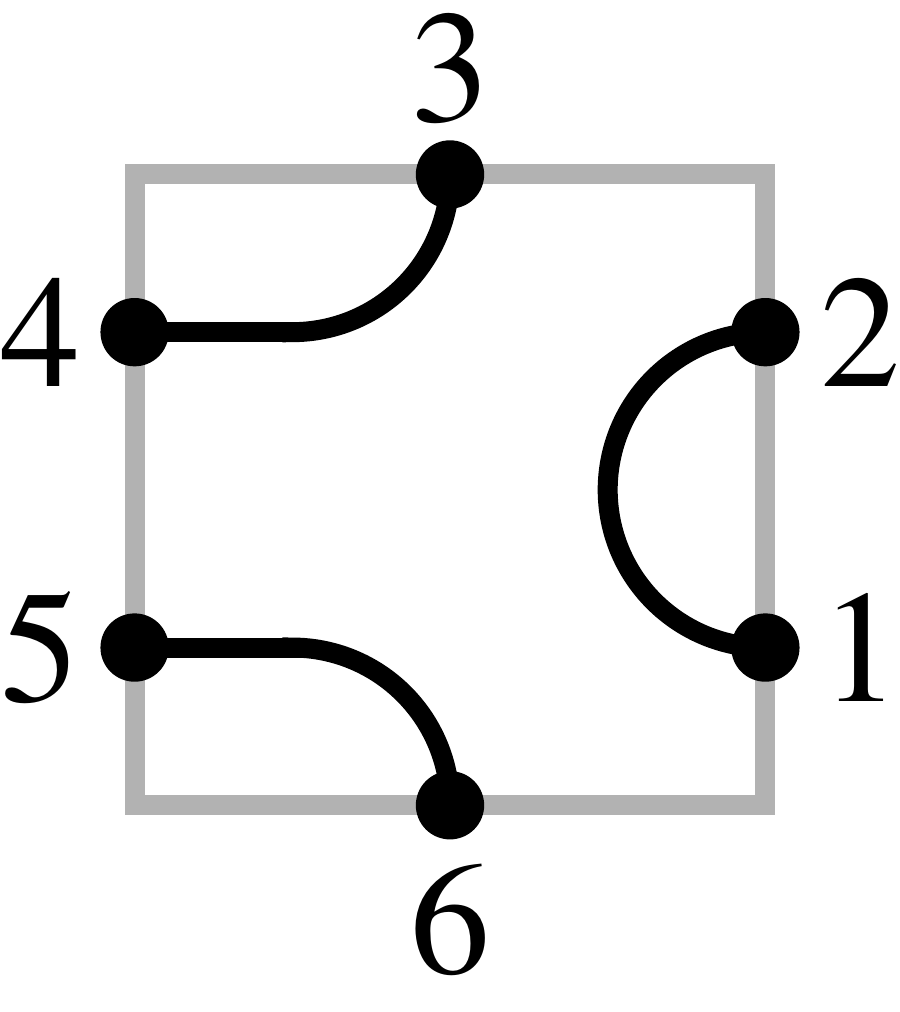}} &
\scalebox{0.18}{\includegraphics{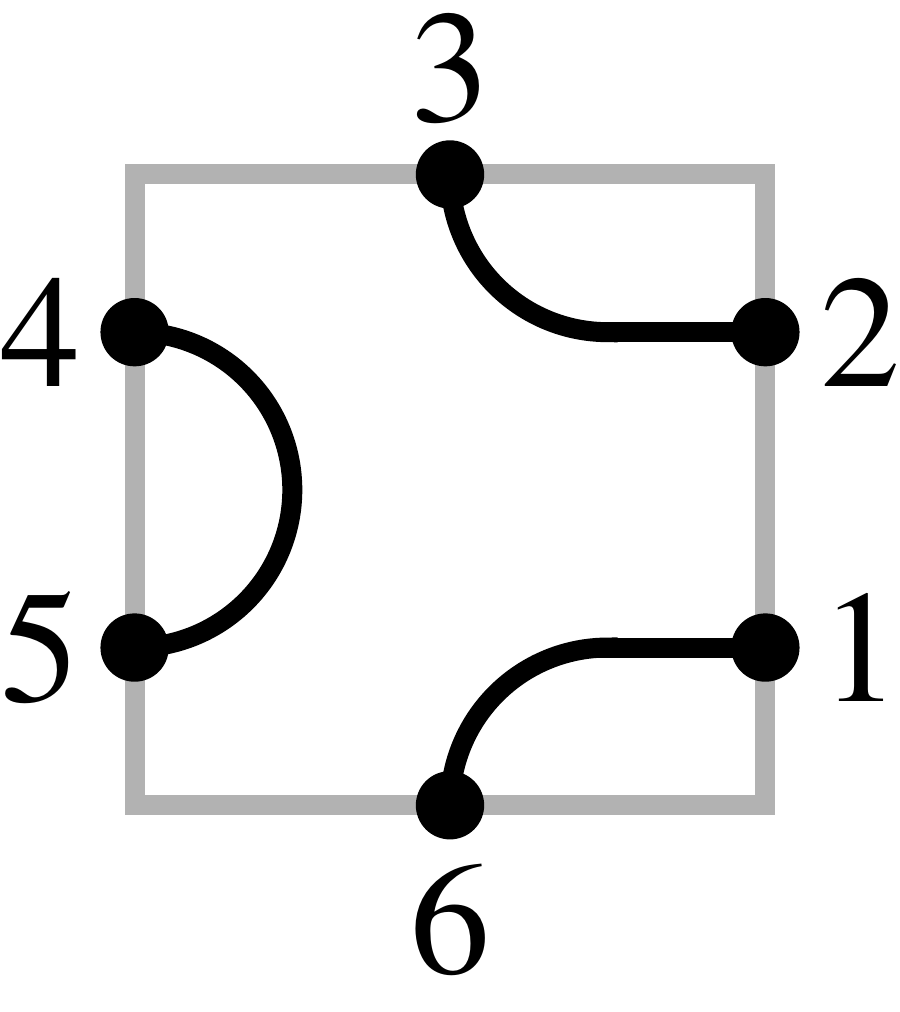}}
\\ $\pi_1$ & $\pi_2$ & $\pi_3$ & $\pi_4$ & $\pi_5$
\end{tabular}
\end{center}
A quick computation shows that the respective probabilities of $\pi_1,\ldots,\pi_5$ are
\begin{align*}
a_1 &= pqs, \\
a_2 &= (1-p)(1-q)s, \\
a_3 &= (1-p)q(1-s), \\
a_4 &= (1-p)qs, \\
a_5 &= pq(1-s) + p(1-q)s + p(1-q)(1-s) + (1-p)(1-q)(1-s),
\end{align*}
in the case of the first arrangement of random plaquettes, and
\begin{align*}
b_1 &= pqs, \\
b_2 &= (1-p)(1-q)s, \\
b_3 &= (1-p)q(1-s), \\
b_4 &= pq(1-s) + p(1-q)s + p(1-q)(1-s) + (1-p)(1-q)(1-s), \\
b_5 &= (1-p)qs,
\end{align*}
for the second arrangement. It follows that the requirement that $a_k=b_k$, $k=1,\ldots,5$, can be satisfied provided that the equation
$$
(1-p)qs = pq(1-s) + p(1-q)s + p(1-q)(1-s) + (1-p)(1-q)(1-s)
$$
holds. Solving for $s$ gives the unique solution
\begin{equation*}
s = s(p,q) = \frac{1-q+pq}{1-p+pq} \in [0,1].
\qedhere
\end{equation*}
\end{proof}

Note that the condition $p\le q$ in the lemma is required to ensure that $s(p,q)$ is a probability. However, it is clear that in the case $p>q$ the lemma remains valid as an algebraic (rather than probabilistic) claim, whereby the stated equality in distribution is interpreted as an equality between two \emph{signed} measures, as long as $1-p+pq\neq 0$.

A second observation required for the application of the Yang-Baxter technique is the following trivial claim.

\begin{lem}
\label{lem:trivial}
Given numbers $s,t\ge 0$, we have the equality in distribution (again, in the sense of signed measures when $s$ or $t$ are not probabilities) of the connectivity patterns of the random plaquette arrangements
\begin{center}
\setlength{\unitlength}{0.4pt}
\begin{picture}(100,140)(0,-20)
\put(0,0){\framebox(50,100){$s$}}
\put(50,0){\framebox(50,100){$t$}}
\put(0,25){\circle*{8}}
\put(0,75){\circle*{8}}
\put(100,25){\circle*{8}}
\put(100,75){\circle*{8}}
\put(-20,20){$4$}
\put(-20,70){$3$}
\put(110,20){$1$}
\put(110,70){$2$}
\end{picture}
\quad \qquad
\raisebox{33pt}{and}
\quad \qquad
\begin{picture}(50,140)(0,-20)
\put(0,0){\framebox(50,100){$st$}}
\put(0,25){\circle*{8}}
\put(0,75){\circle*{8}}
\put(50,25){\circle*{8}}
\put(50,75){\circle*{8}}
\put(-20,20){$4$}
\put(-20,70){$3$}
\put(60,20){$1$}
\put(60,70){$2$}
\end{picture}
\end{center}
\end{lem}

With this preparation, we are ready to prove:

\begin{thm}[Row-switching invariance]
\label{thm:row-switching-inv}
Given probabilities $p,q\in [0,1]$ and an integer $L=2n$, the connectivity pattern of the arrangement of random plaquettes
\begin{center}
\setlength{\unitlength}{0.45pt}
\begin{picture}(500,140)(0,-20)
\put(0,50){\framebox(50,50){$p$}}
\put(50,50){\framebox(50,50){$p$}}
\put(100,50){\framebox(50,50){$p$}}
\put(150,50){\framebox(50,50){$p$}}
\put(0,0){\framebox(50,50){$q$}}
\put(50,0){\framebox(50,50){$q$}}
\put(100,0){\framebox(50,50){$q$}}
\put(150,0){\framebox(50,50){$q$}}
\put(200,0){\framebox(200,100){$\cdots$}}
\put(400,0){\framebox(50,50){$q$}}
\put(450,0){\framebox(50,50){$q$}}
\put(400,50){\framebox(50,50){$p$}}
\put(450,50){\framebox(50,50){$p$}}
\put(25,0){\circle*{8}}
\put(75,0){\circle*{8}}
\put(125,0){\circle*{8}}
\put(175,0){\circle*{8}}
\put(425,0){\circle*{8}}
\put(475,0){\circle*{8}}
\put(25,100){\circle*{8}}
\put(75,100){\circle*{8}}
\put(125,100){\circle*{8}}
\put(175,100){\circle*{8}}
\put(425,100){\circle*{8}}
\put(475,100){\circle*{8}}
\put(20,-25){$1$}
\put(70,-25){$2$}
\put(120,-25){$3$}
\put(170,-25){$4$}
\put(395,-25){$L-1$}
\put(470,-25){$L$}
\put(20,110){$1'$}
\put(70,110){$2'$}
\put(120,110){$3'$}
\put(170,110){$4'$}
\put(380,110){$(L-1)'$}
\put(470,110){$L'$}
\end{picture}
\end{center}
is equal in distribution to the connectivity pattern of the symmetric arrangement
\begin{center}
\setlength{\unitlength}{0.45pt}
\begin{picture}(500,140)(0,-20)
\put(0,50){\framebox(50,50){$q$}}
\put(50,50){\framebox(50,50){$q$}}
\put(100,50){\framebox(50,50){$q$}}
\put(150,50){\framebox(50,50){$q$}}
\put(0,0){\framebox(50,50){$p$}}
\put(50,0){\framebox(50,50){$p$}}
\put(100,0){\framebox(50,50){$p$}}
\put(150,0){\framebox(50,50){$p$}}
\put(200,0){\framebox(200,100){$\cdots$}}
\put(400,0){\framebox(50,50){$p$}}
\put(450,0){\framebox(50,50){$p$}}
\put(400,50){\framebox(50,50){$q$}}
\put(450,50){\framebox(50,50){$q$}}
\put(25,0){\circle*{8}}
\put(75,0){\circle*{8}}
\put(125,0){\circle*{8}}
\put(175,0){\circle*{8}}
\put(425,0){\circle*{8}}
\put(475,0){\circle*{8}}
\put(25,100){\circle*{8}}
\put(75,100){\circle*{8}}
\put(125,100){\circle*{8}}
\put(175,100){\circle*{8}}
\put(425,100){\circle*{8}}
\put(475,100){\circle*{8}}
\put(20,-25){$1$}
\put(70,-25){$2$}
\put(120,-25){$3$}
\put(170,-25){$4$}
\put(395,-25){$L-1$}
\put(470,-25){$L$}
\put(20,110){$1'$}
\put(70,110){$2'$}
\put(120,110){$3'$}
\put(170,110){$4'$}
\put(380,110){$(L-1)'$}
\put(470,110){$L'$}
\end{picture}
\end{center}
where both arrangements are interpreted as \emph{circular} arrangements (i.e., the left and right boundary edges are identified).
\end{thm}

\begin{proof} We may assume without loss of generality that $p\le q$, and furthermore assume that $(p,q)\neq (0,1)$, which implies $s(p,q)>0$ (the claim is trivially true in the case $p=0, q=1$).
Starting with the first arrangement of plaquettes, augment it by inserting at the left-hand side two auxiliary plaquettes with biases $s=s(p,q)$ and $1/s$. This gives the picture
\begin{center}
\setlength{\unitlength}{0.45pt}
\begin{picture}(500,150)(-50,-20)
\put(-100,0){\framebox(50,100){$\displaystyle \frac1s$}}
\put(-50,0){\framebox(50,100){$s$}}
\put(0,50){\framebox(50,50){$p$}}
\put(50,50){\framebox(50,50){$p$}}
\put(100,50){\framebox(50,50){$p$}}
\put(150,50){\framebox(50,50){$p$}}
\put(0,0){\framebox(50,50){$q$}}
\put(50,0){\framebox(50,50){$q$}}
\put(100,0){\framebox(50,50){$q$}}
\put(150,0){\framebox(50,50){$q$}}
\put(200,0){\framebox(200,100){$\cdots$}}
\put(400,0){\framebox(50,50){$q$}}
\put(450,0){\framebox(50,50){$q$}}
\put(400,50){\framebox(50,50){$p$}}
\put(450,50){\framebox(50,50){$p$}}
\put(25,0){\circle*{8}}
\put(75,0){\circle*{8}}
\put(125,0){\circle*{8}}
\put(175,0){\circle*{8}}
\put(425,0){\circle*{8}}
\put(475,0){\circle*{8}}
\put(25,100){\circle*{8}}
\put(75,100){\circle*{8}}
\put(125,100){\circle*{8}}
\put(175,100){\circle*{8}}
\put(425,100){\circle*{8}}
\put(475,100){\circle*{8}}
\put(20,-25){$1$}
\put(70,-25){$2$}
\put(120,-25){$3$}
\put(170,-25){$4$}
\put(395,-25){$L-1$}
\put(470,-25){$L$}
\put(20,110){$1'$}
\put(70,110){$2'$}
\put(120,110){$3'$}
\put(170,110){$4'$}
\put(380,110){$(L-1)'$}
\put(470,110){$L'$}
\end{picture}
\end{center}
which by Lemma~\ref{lem:trivial} does not affect the distribution of the connectivity pattern. Then, perform a sequence of moves pushing the $s$-biased auxiliary plaquette to the right, in accordance with the operation described in Lemma~\ref{lem:version-yang-baxter}. This gives a sequence of random plaquette arrangements, all with distributionally equal connectivity patterns, as follows:
\begin{center}
\begin{tabular}{cc}
&
\setlength{\unitlength}{0.45pt}
\begin{picture}(500,140)(-50,-20)
\put(-100,0){\framebox(50,100){$\displaystyle \frac1s$}}
\put(0,0){\framebox(50,100){$s$}}
\put(-50,50){\framebox(50,50){$q$}}
\put(50,50){\framebox(50,50){$p$}}
\put(100,50){\framebox(50,50){$p$}}
\put(150,50){\framebox(50,50){$p$}}
\put(-50,0){\framebox(50,50){$p$}}
\put(50,0){\framebox(50,50){$q$}}
\put(100,0){\framebox(50,50){$q$}}
\put(150,0){\framebox(50,50){$q$}}
\put(200,0){\framebox(200,100){$\cdots$}}
\put(400,0){\framebox(50,50){$q$}}
\put(450,0){\framebox(50,50){$q$}}
\put(400,50){\framebox(50,50){$p$}}
\put(450,50){\framebox(50,50){$p$}}
\put(-25,0){\circle*{8}}
\put(75,0){\circle*{8}}
\put(125,0){\circle*{8}}
\put(175,0){\circle*{8}}
\put(425,0){\circle*{8}}
\put(475,0){\circle*{8}}
\put(-25,100){\circle*{8}}
\put(75,100){\circle*{8}}
\put(125,100){\circle*{8}}
\put(175,100){\circle*{8}}
\put(425,100){\circle*{8}}
\put(475,100){\circle*{8}}
\put(-30,-25){$1$}
\put(70,-25){$2$}
\put(120,-25){$3$}
\put(170,-25){$4$}
\put(395,-25){$L-1$}
\put(470,-25){$L$}
\put(-30,110){$1'$}
\put(70,110){$2'$}
\put(120,110){$3'$}
\put(170,110){$4'$}
\put(380,110){$(L-1)'$}
\put(470,110){$L'$}
\end{picture}
\\[15pt]
\raisebox{33pt}{$\rightarrow$} \hspace{10pt} &
\setlength{\unitlength}{0.45pt}
\begin{picture}(500,140)(-50,-20)
\put(-100,0){\framebox(50,100){$\displaystyle \frac1s$}}
\put(50,0){\framebox(50,100){$s$}}
\put(-50,50){\framebox(50,50){$q$}}
\put(0,50){\framebox(50,50){$q$}}
\put(100,50){\framebox(50,50){$p$}}
\put(150,50){\framebox(50,50){$p$}}
\put(-50,0){\framebox(50,50){$p$}}
\put(0,0){\framebox(50,50){$p$}}
\put(100,0){\framebox(50,50){$q$}}
\put(150,0){\framebox(50,50){$q$}}
\put(200,0){\framebox(200,100){$\cdots$}}
\put(400,0){\framebox(50,50){$q$}}
\put(450,0){\framebox(50,50){$q$}}
\put(400,50){\framebox(50,50){$p$}}
\put(450,50){\framebox(50,50){$p$}}
\put(-25,0){\circle*{8}}
\put(25,0){\circle*{8}}
\put(125,0){\circle*{8}}
\put(175,0){\circle*{8}}
\put(425,0){\circle*{8}}
\put(475,0){\circle*{8}}
\put(-25,100){\circle*{8}}
\put(25,100){\circle*{8}}
\put(125,100){\circle*{8}}
\put(175,100){\circle*{8}}
\put(425,100){\circle*{8}}
\put(475,100){\circle*{8}}
\put(-30,-25){$1$}
\put(20,-25){$2$}
\put(120,-25){$3$}
\put(170,-25){$4$}
\put(395,-25){$L-1$}
\put(470,-25){$L$}
\put(-30,110){$1'$}
\put(20,110){$2'$}
\put(120,110){$3'$}
\put(170,110){$4'$}
\put(380,110){$(L-1)'$}
\put(470,110){$L'$}
\end{picture}
\\[15pt]
\raisebox{33pt}{$\rightarrow$} \hspace{10pt} &
\setlength{\unitlength}{0.45pt}
\begin{picture}(500,140)(-50,-20)
\put(-100,0){\framebox(50,100){$\displaystyle \frac1s$}}
\put(100,0){\framebox(50,100){$s$}}
\put(-50,50){\framebox(50,50){$q$}}
\put(0,50){\framebox(50,50){$q$}}
\put(50,50){\framebox(50,50){$q$}}
\put(150,50){\framebox(50,50){$p$}}
\put(-50,0){\framebox(50,50){$p$}}
\put(0,0){\framebox(50,50){$p$}}
\put(50,0){\framebox(50,50){$p$}}
\put(150,0){\framebox(50,50){$q$}}
\put(200,0){\framebox(200,100){$\cdots$}}
\put(400,0){\framebox(50,50){$q$}}
\put(450,0){\framebox(50,50){$q$}}
\put(400,50){\framebox(50,50){$p$}}
\put(450,50){\framebox(50,50){$p$}}
\put(-25,0){\circle*{8}}
\put(25,0){\circle*{8}}
\put(75,0){\circle*{8}}
\put(175,0){\circle*{8}}
\put(425,0){\circle*{8}}
\put(475,0){\circle*{8}}
\put(-25,100){\circle*{8}}
\put(25,100){\circle*{8}}
\put(75,100){\circle*{8}}
\put(175,100){\circle*{8}}
\put(425,100){\circle*{8}}
\put(475,100){\circle*{8}}
\put(-30,-25){$1$}
\put(20,-25){$2$}
\put(70,-25){$3$}
\put(170,-25){$4$}
\put(395,-25){$L-1$}
\put(470,-25){$L$}
\put(-30,110){$1'$}
\put(20,110){$2'$}
\put(70,110){$3'$}
\put(170,110){$4'$}
\put(380,110){$(L-1)'$}
\put(470,110){$L'$}
\end{picture}
\\[15pt]
\hspace{-30pt} \raisebox{33pt}{$\rightarrow \cdots \rightarrow$} \hspace{15pt} &
\setlength{\unitlength}{0.45pt}
\begin{picture}(500,140)(0,-20)
\put(500,0){\framebox(50,100){$s$}}
\put(-50,0){\framebox(50,100){$\displaystyle \frac1s$}}
\put(0,50){\framebox(50,50){$q$}}
\put(50,50){\framebox(50,50){$q$}}
\put(100,50){\framebox(50,50){$q$}}
\put(150,50){\framebox(50,50){$q$}}
\put(0,0){\framebox(50,50){$p$}}
\put(50,0){\framebox(50,50){$p$}}
\put(100,0){\framebox(50,50){$p$}}
\put(150,0){\framebox(50,50){$p$}}
\put(200,0){\framebox(200,100){$\cdots$}}
\put(400,0){\framebox(50,50){$p$}}
\put(450,0){\framebox(50,50){$p$}}
\put(400,50){\framebox(50,50){$q$}}
\put(450,50){\framebox(50,50){$q$}}
\put(25,0){\circle*{8}}
\put(75,0){\circle*{8}}
\put(125,0){\circle*{8}}
\put(175,0){\circle*{8}}
\put(425,0){\circle*{8}}
\put(475,0){\circle*{8}}
\put(25,100){\circle*{8}}
\put(75,100){\circle*{8}}
\put(125,100){\circle*{8}}
\put(175,100){\circle*{8}}
\put(425,100){\circle*{8}}
\put(475,100){\circle*{8}}
\put(20,-25){$1$}
\put(70,-25){$2$}
\put(120,-25){$3$}
\put(170,-25){$4$}
\put(395,-25){$L-1$}
\put(470,-25){$L$}
\put(20,110){$1'$}
\put(70,110){$2'$}
\put(120,110){$3'$}
\put(170,110){$4'$}
\put(380,110){$(L-1)'$}
\put(470,110){$L'$}
\end{picture}
\end{tabular}
\end{center}
The connectivity pattern of the last arrangement is equal in distribution to that of the second arrangement in the theorem, again by Lemma~\ref{lem:trivial}.
\end{proof}

Theorem~\ref{thm:row-switching-inv} now easily implies Theorem~\ref{thm:matrices-commute} using similar reasoning to that used in Section~\ref{sec:proof-commutation}. Note however that our pattern-preserving involution $V$ also gives an immediate proof of Theorem~\ref{thm:row-switching-inv} by matching the two-row configurations $\binom{\rho_2}{\rho_1} \in \rows^2$ into pairs according to the rule $\binom{\rho_2'}{\rho_1'}= V \binom{\rho_2}{\rho_1}$, in a way that preserves the connectivity pattern and precisely maps a $(p,q)$-biased probability distribution of plaquettes ($p$-biased plaquettes in the bottom row, $q$-biased in the top row) into the $(q,p)$-biased distribution. On the other hand, the algebraic manipulations in the proof \linebreak above---although quite elegant---have no such combinatorial interpretation. In fact, keeping track of the transformation of (signed!)\ measures at any step during the proof one only gets a summation identity relating \emph{sums} of $(p,q)$-biased probabilities over a certain class of two-row configurations (the ones having a given connectivity pattern of the $2L$ endpoints) to the sum of $(q,p)$-biased probabilities over the same class. From an analysis of the proof above it is not at all clear that such a summation identity can be refined into a pointwise equality between probabilities of pairs of configurations (moreover, the role of the auxiliary parameter $s(p,q)$ is also quite mysterious). Our construction of the involution $V$ provides such a refinement, and therefore gives a somewhat more satisfying explanation of the commutation property of the transfer matrices.

\section{Final remarks}

\label{sec:final-remarks}

\begin{enumerate}

\item \textbf{The invariance theorem.} The path from Theorem~\ref{thm:matrices-commute} to Theorem~\ref{thm:invariance} is simple and relies on well-known ideas. Here is a sketch of the proof. First, note that the technical condition $\sum_j p_j^n (1-p_j)^n=\infty$ implies using the Borel-Cantelli lemma that almost surely there will be infinitely many rows of the form $(\rightplaq \ \leftplaq \ \rightplaq\  \leftplaq\  \ldots\  \rightplaq\  \leftplaq)$; it is easy to see that even a single such row forces the connectivity pattern to be well-defined (i.e., paths cannot escape to infinity). Denote by $D(p_1,p_2,\ldots)$ the distribution of the connectivity pattern. The commutation property of the transfer matrices implies that $D(\ldots p_j,p_{j+1} \ldots) =
D(\ldots p_{j+1},p_j \ldots)$. But it is also clear that $D(\cdot)$ depends only weakly on the far-away coordinates $p_j$ with $j$ large, so exploiting this (through a suitable limiting argument that is left to the reader), a row with some fixed bias, say $r_0=1/2$, can be ``brought from infinity'' to show that
\begin{align*}
D(p_1,p_2,\ldots) &= D(p_1,p_2,\ldots,p_{10^{10000}}, \ldots, r_0 \textrm{ [infinitely far away]} )
\\ &=
D(p_1,p_2,\ldots,p_j, r_0, p_{j+1}, \ldots) = \ldots = D(r_0, p_1, p_2, \ldots).
\end{align*}
By induction, we get that
$D(p_1,p_2,\ldots) = D(\,\overbrace{r_0,\ldots,r_0}^{m\textrm{ times}}\,, p_1,p_2,\ldots)
$
for every $m\ge 1$, and another limiting argument then implies that $D(p_1,p_2,\ldots)=D(r_0,r_0,\ldots)$.

\item \textbf{Combinatorial approach to other applications of the Yang-Baxter equation.} The Yang-Baxter equation is an important tool with deep consequences in statistical mechanics \cite{baxter}, algebraic and enumerative combinatorics \cite{bressoud, kuperberg1, kuperberg2}, and knot theory \cite{wu}. In many of these applications, the Yang-Baxter equation is used to reveal subtle symmetries of a problem that are hard to detect using other methods. In this paper we showed however that a probabilistic symmetry that was previously proved using this somewhat mysterious algebraic device has a direct combinatorial explanation. This raises the intriguing possibility that some of the other important applications of the Yang-Baxter equation can be approached using similar ideas.

\end{enumerate}

\section*{Acknowledgements}

The authors thank Omer Angel, Gady Kozma, Gidi Amir, Greg Kuperberg and Roger Behrend for helpful discussions. Dan Romik was supported by the National Science Foundation under grant DMS-0955584, and by grant \#228524 from the Simons Foundation. Part of the work on this paper was done during his visit to the Erwin Schr\"odinger International Institute for Mathematical Physics, and he is grateful for the Institute's support. Ron Peled was supported by the Israeli Science Foundation grant number 1048/11 and by the Marie Curie IRG grant SPTRF.

\bigskip

\noindent
\begin{tabular}{lcl}
Ron Peled & & Dan Romik
\\
School of Mathematical Sciences
&&Department of Mathematics \\
Tel-Aviv University
&&University of California, Davis \\
Tel-Aviv 69978
&&One Shields Ave, Davis, CA 95616 \\
Israel && USA
\\[6pt]
Email: \texttt{peledron@post.tau.ac.il}
&&Email: \texttt{romik@math.ucdavis.edu}
\end{tabular}

\end{document}